\documentclass[twocolumn]{autart}    % Enable this line and disable the 
\usepackage{graphicx}          % Include this line if your 
\usepackage{bm}
\usepackage{amssymb}
\usepackage{apacite}
\usepackage{wrapfig}

\usepackage{natbib}
\usepackage{enumerate}
\usepackage{mathrsfs}
\usepackage{amsmath}
\usepackage{mathtools}
\usepackage[table]{xcolor}
\usepackage[colorlinks,citecolor=black!80,urlcolor=black,bookmarks=false,hypertexnames=true]{hyperref} 

\newcommand{\reals}{\mathbb{R}}

\newcommand{\integers}{\mathbb{Z}}

\newcommand{\naturals}{\mathbb{N}}

\newcommand{\T}{\top}

\newcommand*\xbar[1]{%
  \hbox{%
     \vbox{%
      \hrule height 0.7pt % The actual bar
      \kern0.35ex%         % Distance between bar and symbol
      \hbox{%
         \kern-0.1em%      % Shortening on the left side
         \ensuremath{#1}%
         \kern-0.1em%      % Shortening on the right side
      }%
     }%
  }%
} 

\newenvironment{proof}{\textbf{Proof.~}}{\hfill$\square$}

\allowdisplaybreaks

\begin{document}

\begin{frontmatter}

% \title{Enhanced prediction for discrete-time input-delayed systems with unknown disturbances\thanksref{footnoteinfo}}

% \title{\textcolor{black}{Predicting the future state of disturbed input-delayed systems based on the calculus of finite differences}\thanksref{footnoteinfo}}

% \title{\textcolor{black}{High-order observer-predictor for disturbed input-delayed systems based on the calculus of finite differences}\thanksref{footnoteinfo}}

\title{Newton-series-based observer-predictor control for disturbed input-delayed discrete-time systems\thanksref{footnoteinfo}}

% \title{\textcolor{black}{High-order observer-predictor for input-delayed systems with unknown disturbances: discrete-time case}\thanksref{footnoteinfo}}

% \title{\textcolor{black}{Predictive-based control of saturated} discrete-time input-delayed systems with unknown disturbances\thanksref{footnoteinfo}}

% \title{Regional predictive-based input-to-state stabilization of discrete-time input-delayed systems with actuator saturation \thanksref{footnoteinfo}}
% Title, preferably not more 

\thanks[footnoteinfo]{This paper was not presented at any conference. Thanks to the Brazilian agency CAPES for financial support.\\ $^\ast$ Corresponding author. \\
\phantom{aa}\textit{Email addresses:} \textbf{\scriptsize thiago.alveslima@uclouvain.be} (T. A. Lima), \textbf{\scriptsize valessa@alu.ufc.br} (V. V. Viana), \textbf{\scriptsize bismark@dee.ufc.br} (B. C. Torrico), \textbf{\scriptsize fnogueira@dee.ufc.br} (F. G. Nogueira), \textbf{\scriptsize dmadeira@dee.ufc.br} (D. de S. Madeira).}

\vspace{-0.5cm}
\author[UC]{Thiago Alves Lima*}, \author[UFC]{Valessa V. Viana}, \author[UFC]{Bismark C. Torrico},  \author[UFC]{Fabrício G. Nogueira},  \author[UFC]{Diego de S. Madeira} 

\address[UC]{ICTEAM, UCLouvain, 4 Av. G. Lema\^itre, 1348 Louvain-la-Neuve, Belgium}

\address[UFC]{Department of Electrical Engineering, Federal University of Ceará, Fortaleza, CE, Brazil}  % Please supply   

          \vspace{-0.5cm}
                             % keyword list or with the 
                                          % help of the Automatica 
                                          % keyword wizard

\vspace{-0.3cm}
\begin{abstract}                          % Abstract of not more than 200 words.
This paper deals with the problem of predicting the future state of discrete-time input-delayed systems in the presence of unknown disturbances that can affect both the state and the output equations of the plant. Since the disturbance is unknown, computing an exact prediction of the future plant states is not possible. To circumvent this problem, we propose using a high-order extended Luenberger-type observer for the plant states, disturbances, and their finite difference variables, combined with a new equation for computing the prediction based on Newton's series from the calculus of finite differences. Detailed performance analysis is carried out to show that, under certain assumptions, both enhanced prediction and improved attenuation of the unknown disturbances are achieved. Linear matrix inequalities (LMIs) are employed for the observer design to minimize the prediction errors. A stabilization procedure based on an iterative design algorithm is also presented for the case where the plant is affected by time-varying uncertainties. Examples from the literature illustrate the advantages of the scheme.
\end{abstract}

\end{frontmatter}

\section{Introduction}

Time delays have been extensively studied over the years due to their harmful impact on the closed loop, which can cause undesired oscillatory behavior or even instability. Due to being infinite-dimensional systems (in the continuous-time case), the analysis and control of time-delayed plants are evolved when compared to non-delayed ones \citep{Fridman_2014,Gu2003}.

When controlling such systems, a well-established idea is that of predicting the system output or state ahead of the delay and then to feedback such prediction so that the closed-loop system is equivalent to that of a non-delayed system. In such cases, one can say that the delay has been \textit{compensated}, that is, its undesired effects have been mitigated by means of prediction. This idea was first developed in frequency domain for single-input single-output (SISO) open-loop stable systems by the seminal work of \citet{Smith_1957} and was later extended to multiple-input multiple-output (MIMO) open-loop unstable systems with the aid of time-domain analysis \citep{Manitius_1979,Artstein_1982}. Recently, interest in the problem of prediction has resurfaced in the context of unknown disturbances. In case the system is affected by such disturbances, it becomes impossible to perfectly predict the future states of the system due to the solution being dependent on future values of the disturbance. Some important works in the last years have tackled this problem. 

To deal with unknown disturbances, \citet{LECHAPPE2015179} proposed a solution based on a modification to the Artstein predictor \citep{Artstein_1982}. The main idea consisted of adding a term that compares the current state of the plant with the delayed prediction, which leads to some information about the disturbance being feedbacked into the scheme. Both \citet{SANZ2016205} and \citet{CASTILLO2021} employ high-order extended observers capable of estimating the disturbance and its derivatives up to some order $r$. Then, such observations are used to help define new predictive schemes that lead to a decrease in the error caused by the unknown disturbance. In \citet{Furtat_2018}, the Finite Spectrum Assignment (FSA) idea from \citet{Manitius_1979} is used along with a disturbance predictor to propose a new control law for disturbance compensation for input-delayed systems. In \citet{HAO2019}, the idea of sequential predictors was introduced, where the need to store past values of the control input is relinquished. More recently, \citet{Wu_2021} tackled the problem of predicting the states of discrete-time input-delayed systems in the case of unknown disturbances by proposing two modifications to the idea originally published for continuous-time systems in \citet{LECHAPPE2015179}. In \citet{GONZALEZ2021}, new advances were reported concerning the consideration of time-varying delays.

In this same vein, in this work we aim at proposing a new prediction scheme for input-delayed systems affected by unknown disturbances that can affect both the state and output equations of the plant. Differently from \citet{Wu_2021}, direct measurement to the plant state is not assumed to be available, which is a situation commonly found in industrial applications. Inspired by \citet{CASTILLO2021}, we employ a high-order extended state observer which is capable of estimating the disturbance and its finite differences up to some order $r$. The main idea consists then of plugging such estimations into a truncated version of the Newton series (from the calculus of finite differences) to estimate future disturbances, which are then used into the solution of the plant to generate a new predictive scheme. We demonstrate that, when mild assumptions are met, such scheme is capable of generating enhanced results with respect to both disturbance attenuation and minimization of prediction errors compared with the recent literature. We also develop conditions for the robust design of the predictive controller scheme in the case of time-varying uncertainties affecting the plant.

The main novelties with respect to \cite{CASTILLO2021} are: i) we use the newton-series to develop the new predictive scheme, which had not been employed before; ii) the closed-loop stability and stabilization in the case of plant time-varying uncertainties is developed; iii) detailed analysis of the disturbance attenuation characteristics of the proposed predictive-based control loop is presented using the reduction approach, while we show that perfect attenuation can be achieved for a big class of disturbances with a simple modification in the control law. The advantages of the strategy are also demonstrated in numerical examples borrowed from \cite{HAO2019} and \cite{GONZALEZ2021}.

% The rest of this paper is organized as follows. We start by providing a general formulation and the problem statement in Section \ref{sec:ProblemFormulation}. Next, in Section \ref{sec:PredictionScheme}, we present the employed extended observer and derive the equations for the proposed predictive scheme. Moving forward, in Section \ref{sec:mainresults}, we present the results concerning analysis and design of the predictive scheme. Section \ref{sec:dist} presents disturbance attenuation analysis and Lyapunov redesign for the uncertain case. Numerical examples from the literature are then presented in Section \ref{sec:numericalexample}. Finally, in Section \ref{sec:conclusions}, we end the paper with concluding remarks and future perspectives. 

\textbf{Notation.} For matrices $W$ and $Z$ in $ \reals^{n \times n}$, $W \succ Z$ means that $W-Z$ is positive definite. \textit{diag}$(W,Z)$ corresponds to the block-diagonal matrix. $\mathbb{S}_n^{+}$ stands for the set of symmetric positive definite matrices. $I$ and $0$ denote identity and null matrices. The symbol $\star$ denotes symmetric blocks in the expression of a matrix. For integers $a < b$, we use $[a,b]$ to denote the set $\{a,a+1,\dots,b-1,b\}$. For a discrete function $f(k) : \integers \rightarrow \reals^{n}$, the forward difference operator is defined by $\Delta f(k)=f(k+1)-f(k)$. Similarly, $\Delta^2 f(k)= \Delta f(k+1) - \Delta f(k)$, $\dots$, $\Delta^{r+1} f(k)= \Delta^r f(k+1) - \Delta^r f(k)$. Furthermore, we define $\Delta^0 f(k) = f(k)$. Additionally, the $d$ steps backwards different is denoted as $\nabla_d f(k) = f(k)-f(k-d)$. We use the notation $\|f(k)\|$ to denote the vector norm $\sqrt{f^{\T}(k)f(k)}$, whereas $\| f \|_{l_2^{loc}}$ is used to denote the local $l_2$ norm of $f(k)$, given by $\sqrt{{\sum}_{\tau=0}^{k} \|f(\tau)\|^{2}}$. Finally, $ \binom{x}{n}=\frac{x^{\underline{n}}}{n!}$ denotes the binomial coefficient, where $x^{\underline{n}}$ stands for the falling factorial $x^{\underline{n}}= { \prod_{j=0}^{n-1}}(x-j)$.
\section{Problem formulation}
\label{sec:ProblemFormulation}

\subsection{General view}\label{subsec:generalview}

Consider an input-delayed plant given by
\begin{equation}\label{eq:plant}
    \begin{cases}
    x(k+1)=A_\Omega(k) x(k) + B_\Omega(k) (u(k-d)+w(k)),\\
    y(k)=C x(k) + D_w w(k),\\
    u(k) = \mathcal{u}(k),~k\in[-d,-1], 
    \end{cases}
\end{equation}
where \(x(k) \in \reals^{n_p},~ u(k) \in \reals^{m_u},~w(k) \in \reals^{q}\), and \(y(k) \in \reals^{m_y}\) are the plant state, the control input, the unknown disturbance, and the plant output, respectively. Furthermore, $x(0)$ and $\mathcal{u}(k)$ define the system initial condition. The plant input delay $d \geq 0$ is assumed to be known. $A_{\Omega}(k)=A+\Omega_A(k)$, $B_{\Omega}(k)=B+\Omega_B(k)$, where $\Omega_A(k)$ and $\Omega_B(k)$ are the time-varying uncertainties in the model and $A$, $B$, $C$, $D_w$ are known matrices of appropriate dimensions satisfying the following assumption.
\begin{assum}\label{assump:system}
The pair $\left(A,C\right)$ is observable and
\begin{equation}
    rank\left( \begin{bmatrix}
    A - I_{n_p} & \phantom{-}B \\
    -C & -D_w
    \end{bmatrix}\right) = n_p+q.
\end{equation}
\end{assum} 
Moreover, the time-varying model uncertainties are described as in \cite{GONZALEZ2021}:
\begin{equation}\label{uncertainties}
    (\Omega_A(k),\Omega_B(k))=\lambda E\Omega(k)(F_A,F_B)
\end{equation}
where \(\lambda \geq 0\) is a scalar that determines the size of uncertainties, \(\Omega(k) \in \mathbb{R}^{l_1 \times l_2}\) is an unknown time-varying matrix satisfying \(\Omega(k)^\T \Omega(k) \leq I, ~\forall k \geq 0\), and \(E, F_A, F_B\) are constant matrices.
%\vspace{-0.5cm}

% One classical idea to control time-delay systems dating back to the 50s \citep{Smith_1957} consists of feedbacking a prediction of the output of the equivalent delay-free system (i.e. $d=0$ in \eqref{eq:plant}) so that any controller designed without taking into account the delay would appropriately stabilize the closed-loop system. In the case of state-feedback control and regarding the time-domain model \eqref{eq:plant} this would be equivalent to feedbacking the value of $x(k+d)$, that is $u(k) = K x(k+d)$, so that the following closed-loop equation would be obtained
% \begin{equation}
%     x(k+1)=(A + B K) x(k) + B w(k),~~ k \geq d,
% \end{equation}
% \noindent thus eliminating the delay influence from the closed loop. 

Consider the nominal system, i.e., $\lambda=0$. From recursion in \eqref{eq:plant}, an exact prediction $x(k+d)$ can be found as
\begin{subequations}
  \begin{gather}
\label{eq:exactprediction}   x(k+d)= \hat{x}_1(k+d) + \sum_{j=1}^{d} A^{j-1} B w(k+d-j), \\ 
 \label{eq:classicalprediction}  \hat{x}_1(k+d)=A^d x(k) + {\displaystyle \sum^{d}_{j=1}A^{j-1}}B u(k-j). 
  \end{gather}
\end{subequations}
One can notice, however, that equation \eqref{eq:exactprediction} cannot be computed since it depends on the current and future values of the disturbance, that is the values \(w(k+s),~ s \in [0,d-1]\). Having knowledge of such values is unthinkable in almost all real systems, therefore strategies to compute an approximation of \eqref{eq:exactprediction} have been studied along the years. One classical choice to compute approximated predictions is to ignore the term due to the disturbance in \eqref{eq:exactprediction}, leading to the classical prediction $\hat{x}_1(k+d)$ in \eqref{eq:classicalprediction}. Clearly, \eqref{eq:classicalprediction} results in a prediction error given by
\begin{equation}\label{classicpredic:error}
    x(k+d)-\hat{x}_1(k+d)={\displaystyle\sum_{j=1}^{d}} A^{j-1} B w(k+d-j).
\end{equation}
The recently published strategy in \citet{Wu_2021} showed that such choice is not appropriate due to large prediction errors that lead to poor attenuation of the disturbances. To deal with this, \citet{Wu_2021} proposes two predictors based on the original idea for continuous-time plants presented in \citet{LECHAPPE2015179}, which are given below
\begin{subequations}
  \begin{gather}
    \hat{x}_2(k+d)=\hat{x}_1(k+d)+x(k)-\hat{x}_1(k), \\
    \hat{x}_3(k+d)=\hat{x}_2(k+d)+x(k)-\hat{x}_2(k).
  \end{gather}
\end{subequations}
In this same vein, in this paper we will also deal with the problem of finding a new prediction, namely $\hat{x}_4(k+d)$, that leads to smaller errors and consequently better attenuation of disturbances. To this end, we first make the following assumption on $w(k)$.
\begin{assum}\label{assump:disturbance}
The disturbance $w(k)$ is \((r+1)\)-times finite differentiable with respect to the \(\Delta\) operator so that for any positive integer $r$ 
\begin{align}
\|\Delta^{(r+1)}w(k)\| \leq \delta, 
\end{align}
\noindent which implies that $\| \Delta^{(r+1)}w \|_{l_2^{loc}} \leq \delta\sqrt{k+1}$, 
\noindent for $ 0\leq k < \infty$, which means that $\Delta^{(r+1)}w(k)$ belongs to the $l_2^{loc} (\integers)$ space. 
\end{assum}

% \subsection{Problem statement}

% The following statement summarizes the problem we intend to solve in this paper. 

% \begin{prob}\label{problem1}
% Given the plant matrices $A,B,C,D_w$, the delay $d$, and taking into account Assumption \ref{assump:disturbance}, we aim at computing a prediction $\hat{x}_4(k+d)$ with the following characteristics:

% \begin{enumerate}[(i)]
% \item For all $k \in [0,\infty)$ the prediction error is limited and bounds on its $l_2$ norm, given by
%     \begin{equation*}
%         \sqrt{\displaystyle \sum_{\tau=0}^{k} \|x(\tau+d)-\hat{x}(\tau+d)\|^2},
%     \end{equation*}
%     \noindent can be analytically expressed. 
%     \item The prediction error can be minimized through LMI-based design of the prediction scheme.
%     \item The set of disturbance signals $w(k)$ for which the prediction leads to a tighter ultimate bound on the plant states is characterized. 
%     \item Robust stabilization is achieved in the case of time-varying uncertainties in the plant. 
% \end{enumerate}
% \end{prob}

\section{Prediction scheme}
\label{sec:PredictionScheme}
In this section we present a new predictive scheme for system \eqref{eq:plant}. The main idea consists of employing a high-order extended observer that allows to estimate the disturbances and their finite differences $\Delta$ up to order $r$. Such observations are then plugged into a series to approximate the future values of the unknown disturbance, leading to the computation of the predictions. We start the section by defining the high-order extended observer equation, followed by presentation of the proposed predictive scheme. 

\subsection{High-order extended state observer}

Since system $\eqref{eq:plant}$ does not give direct measurement to the plant state, a first step to compute a prediction is to estimate the value of $x(k)$, which can be done by means of a state observer. Furthermore, to deal with the disturbance summation error term \eqref{classicpredic:error}, it is necessary to gather some knowledge about the disturbance as well, which can help decrease the error caused by this term. Such task can be accomplished by means of an extended state observer. In this work we employ observers that allow to estimate not only the plant state and the disturbance signal but also the difference operators (up to order $r$) of the disturbance. The idea of using high-order extended observers is inspired by \citet{CASTILLO2021}, which deals with continuous-time systems, where the disturbances and their time derivatives are observed. Herein, we propose the use of the following high-order observer to deal with the case of discrete-time systems
\begin{equation} \label{observer}
    \begin{cases}
    \hat{\eta}(k+1)=\xbar{A}
    \hat{\eta}(k) + \xbar{B}_u u(k-d) + L e_y(k),\\
    \hat{y}(k)=\xbar{C}\hat{\eta}(k),
    \end{cases}
\end{equation}
\noindent where  \(\hat{\eta}(k)=\left[\hat{x}(k)^\T ~~ \hat{\bm{w}}^\T(k)\right]^\T \in \reals^n\), with $n=n_p+(r+1)q$, is the observer state, $e_y(k) = y(k)-\hat{y}(k)$ is the observation error, $L \in \reals^{n \times m_y}$ is the observer gain to be designed, $\xbar{C}=\begin{bmatrix}
    C & D_w &  0_{n_p \times rq} 
    \end{bmatrix}$, and
    \vspace{-0.1cm}
\begin{equation*}
   \begin{split}
        &\xbar{A}=
    \begin{bmatrix}
    A & \left[B ~~ 0_{n_p \times rq}\right]\\
    0_{(r+1)q \times n_p} & \Pi
    \end{bmatrix},~~ \xbar{B}_u=\begin{bmatrix}
    B \\
     0_{(r+1)q \times m_u} 
    \end{bmatrix}, \\
    & \Pi = I_{(r+1)q} + \begin{bmatrix}
    I_q & \dots & I_q & 0_{q \times q} 
    \end{bmatrix}^{\T} \begin{bmatrix}
    0_{q \times q} & I_q & \dots & I_q
    \end{bmatrix}.
   \end{split}
\end{equation*}
System \eqref{observer} is a Luenberger-type observer for the augmented variable \(\eta(k)=\left[x(k)^\T ~~ \bm{w}^\T(k)\right]^\T\) $\in \reals^{n}$, with \(\bm{w}^\T(k)=\left[w^\T(k)~\Delta w^\T(k)~\cdots~\Delta^r w^\T(k) \right]^\T\), which under Assumption \ref{assump:disturbance} satisfies ( with $\xbar{B}_w = \begin{bmatrix}
0_{q \times rq+n_p} & I_q \end{bmatrix}^\T$)
\begin{equation}\label{eq:Eta}
    \eta (k+1) = \xbar{A}
    {\eta}(k) + \xbar{B}_u u(k-d) + \xbar{B}_w \Delta^{r+1} w(k).\\
\end{equation}
% \noindent where $\xbar{B}_w = \begin{bmatrix}
% 0_{q \times rq+n_p} & I_q \end{bmatrix}^\T$. 
\vspace{-0.8cm}
\begin{lem}\label{lemma:obs}
Under Assumption \ref{assump:system}, the pair $\left(\xbar{A}, \xbar{C}\right)$ is observable, allowing proper estimation of the variable $\eta(k)$.
\end{lem}
\begin{proof}
% See Appendix \ref{appendix:obs} for proof. 
% The proof follows the same reasoning as the one presented for a low-order disturbance observer in Section III of \citet{Chang_2006}, which employs the Hautus Lemma for observability \citep{sontag1998mathematical}, and thus is omitted.  
% The proof follows from the Hautus Lemma for observability \citep{sontag1998mathematical}, and is omitted due to space constraints.  
The proof follows from the Hautus Lemma for observability \citep{sontag1998mathematical,Chang_2006}, and is omitted due to space constraints.
\end{proof}
% Next, we will present a new strategy employing the observed variables to compute a prediction to system \eqref{eq:plant}. 

\subsection{A new expression for the prediction}

% In an attempt to improve prediction of continuous-time input-delayed systems in the presence of unknown disturbances, both xx and xx apply the well-known Taylor series to make approximations of the future values of the disturbance with observations of the time-derivatives of the disturbance. Inspired by these ideas, in this work we employ a series which is appropriate for the case of discrete-time systems and that uses finite differences instead of derivatives to approximate functions. 

From the calculus of finite differences, the following expression, known as the Newton series, holds for the operator \(\Delta\) \citep{jordan1950calculus,Rota1994}
\begin{equation}\label{eq:sum_w}
    w(k+s)=\sum_{m=0}^{\infty} \binom{s}{m} \Delta^m w(k).
\end{equation}
From equation \eqref{eq:sum_w} we can rewrite \eqref{eq:exactprediction} as 
\begin{equation}\label{eq:exactprediction2}
\begin{split}
    &x(k+d)=A^d x(k) + \sum_{j=1}^{d} A^{j-1}B u(k-j)  \\
    &+ \sum_{j=1}^{d} A^{j-1} B \left[ \sum_{m=0}^{\infty} \binom{d-j}{m} \Delta^m w(k) \right].
    \end{split}
\end{equation}
By splitting the infinite sum in \eqref{eq:exactprediction2} into two parts, one with terms ranging between $0$ and $r$ and other with terms between $r+1$ and $\infty$, we arrive at the following expression for \eqref{eq:exactprediction}
\begin{equation}\label{eq:exactprediction3}
\begin{split}
    &x(k+d)=A^d x(k) + \sum_{j=1}^{d} A^{j-1}B u(k-j) \\ 
    & + \sum_{j=1}^{d} A^{j-1} B 
    \left[ \sum_{m=0}^{r} \binom{d-j}{m} \Delta^m w(k) \right]\\
    &+  \sum_{j=1}^{d} A^{j-1} B 
    \left[ \sum_{m=r+1}^{\infty} \binom{d-j}{m} \Delta^m w(k) \right],
    \end{split}
\end{equation}
\noindent where the second summation term depends on the differences up to order $r$ of the disturbance $w(k)$. Then, given any observation of the variable \(\eta(k)\), the following prediction can be defined  
\begin{equation}\label{eq:predicfirst}
  \begin{split}
    \hat{x}_4(k+&d) = A^d \hat{x}(k) + \sum_{j=1}^{d} A^{j-1}B u(k-j) \\ &+ \sum_{j=1}^{d} A^{j-1} B 
    \left[ \sum_{m=0}^{r} \binom{d-j}{m} \Delta^m \hat{w}(k) \right],
  \end{split}
\end{equation}
\noindent where the summation with terms ranging from $0$ to $r$ is a truncation of the infinite Newton series which can yield approximate estimations for the future values of the disturbance $w(k+s)$, $s \in [1,d-1]$. By rewriting \eqref{eq:predicfirst}, an expression for the prediction depending on the observer state variable $\hat{\eta}(k)$ is given by
\begin{subequations}
\begin{gather}
\label{eq:prediction}
    \hat{x}_4(k+d) = \Gamma(d) \hat{\eta}(k) + \sum_{j=1}^{d} A^{j-1}B u(k-j), \\
    \label{eq:GammaDMain}
    \Gamma(d) = \begin{bmatrix}
A^d & T(d)
\end{bmatrix}, \\
 T(d) \hspace{-0.1cm}=\hspace{-0.1cm} \sum_{j=1}^d \hspace{-0.1cm}A^{j-1} B \hspace{-0.1cm} \begin{bmatrix}
    I_q \binom{d-j}{0} & I_q \binom{d-j}{1} & \cdots I_q \binom{d-j}{r}
    \end{bmatrix}.
    \end{gather}
\end{subequations}
Therefore, in this work, we propose the utilisation of prediction \eqref{eq:prediction}, which can be implemented by using the high-order extended state observer \eqref{observer}. The main advantage of employing \eqref{eq:prediction} in comparison with \eqref{eq:classicalprediction} is that we are able to use information from the disturbance in the prediction even though the disturbance is unknown. This way, the prediction error \eqref{classicpredic:error} can be significantly reduced, as discussed in the next subsection. \textcolor{black}{Let us recall that predicting the disturbance $w: \naturals \to \reals^{m}$ is necessary to improve the prediction of the state since the latter depends on the former, as motivated in Section \ref{subsec:generalview}.} 

% \begingroup
% \color{black} \begin{rem}
% The main motivation to employing the Newton series in this work is that it can be used to approximate discrete-time signals (sequences) as the disturbance $w: \naturals \to \reals^{m}$ in this paper. The employment of other commonly used series as the celebrated Taylor series is impaired in this scenario since the derivative of functions defined on the natural numbers is not well defined, which makes it more adequate to dealing with continuous-time signals as in \cite{SANZ2016205}. %For example, $w(k)=\sqrt{k}$, $k \in \naturals$, cannot be expanded around any point with Taylor's series. In general, the use of the Newton series is superior when dealing with discrete signals, as for example recently advocated in the context of quantum physics \citep{Vogl_2020,Konig2021}.
% \end{rem}
% \endgroup

\section{Nominal predictor analysis and design}\label{sec:mainresults}

In this section, we rigorously analyse the prediction characteristics and employ LMI-based design to the observer \eqref{observer} with the aim to minimize influence of the disturbance in the prediction error. The next subsection brings an initial discussion on the prediction error analysis. 

\subsection{Prediction error analysis}
From \eqref{eq:exactprediction3} and \eqref{eq:prediction}, a unique expression for the prediction error can be found as follows
\begin{subequations}
\begin{gather}
    \label{eq:predictionerror}
    x(k+d)-\hat{x}_4(k+d) = \mathscr{E}_{\hat{o}}(k)+\mathscr{E}_r(k), \\
    \label{eq:Oe} \mathscr{E}_{\hat{o}}(k) = \Gamma(d) (\eta(k)-\hat{\eta}(k)), \\
    \label{eq:Or}
\mathscr{E}_r(k)\hspace{-0.02cm}=\hspace{-0.02cm}\sum_{j=1}^{d} A^{j-1} B \hspace{-0.1cm} \left(\sum_{m=r+1}^{\infty} \hspace{-0.1cm}\binom{d-j}{m} \Delta^m w(k)\hspace{-0.1cm} \right).
\end{gather}
\end{subequations} 
\vspace{-0.5cm}
\begin{prop}
If the error $\mathscr{E}_{\hat{o}}(k)+\mathscr{E}_r(k) \rightarrow 0$, then the asymptotic convergence of the prediction $\hat{x}_4$ to 0 implies the asymptotic convergence of the state $x$ to zero.  
\end{prop}
\begin{proof}
Proof is straightforward from \eqref{eq:predictionerror}. 
\end{proof} 

From \eqref{eq:Oe}, note that $\mathscr{E}_{\hat{o}}(k)$ is an error term that depends on the quality of the observation, which can be minimized by proper design of the observer \eqref{observer}, i.e. by design of the gain $L$. Such design will be realized in Section \ref{sec:designnominal}. On the other hand, error $\mathscr{E}_r(k)$ in Equation \eqref{eq:Or} is an error which is inevitable to the prediction. Nonetheless, under Assumption \ref{assump:disturbance}, the $l_2$ norm of this error is bounded, as stated in the next Proposition.
\begin{prop}\label{propOr}
By taking into account Assumption \ref{assump:disturbance}, the $l_2$ norm of $\mathscr{E}_r(k)$ is bounded such that
\vspace{-0.2cm}
\begin{equation}\label{eq:propOr}
   \|\mathscr{E}_{r}\|_{l_2^{loc}} \leq \sqrt{\sum_{\lambda=0}^{k} \delta^2 d^2 \mu^2}
    = \delta d \mu \sqrt{k+1}, 
\vspace{-0.5cm}
\end{equation}
\(\forall k \in [0,\infty)\),  where $\mu =\displaystyle \max_{j=\{1 \dots d\}} \left(\sigma_{max}(Y_j^{\frac{1}{2}})\phi_j\right)$, with $\phi_j = \displaystyle \sum_{l=0}^{d-j-r-1}  \binom{d-j}{l+r+1} 2^l$, \(Y_j= B^\T {A^{^\T j-1}} A^{j-1}B\), and $\sigma_{max}(Y_j^{\frac{1}{2}})$ is the maximum singular value of $Y_j^{\frac{1}{2}}$. Furthermore, in case $r>d-2$, $\|\mathscr{E}_{r}\|_{l_2^{loc}}=0$, \(\forall k \in [0,\infty)\). 
\end{prop}
\begin{proof}
See Appendix \ref{appendix:Or} for proof. 
\end{proof}

Therefore, even though the disturbance $w(k)$ is unknown, by employing prediction \eqref{eq:prediction} we can guarantee that the error caused by $\mathscr{E}_r(k)$ is limited, as shown by \eqref{eq:propOr}. Furthermore, as cited in Proposition \ref{propOr}, in the special circumstance that $r>d-2$, its $l_2$ norm is null. As long as we know, this kind of relation between the observer parameter $r$ and the time delay $d$ has not been shown to hold in other predictive-based strategies from the literature. One of the objectives of this paper is to show that the $l_2$ norm of the prediction error $ x(k+d)-\hat{x}_4(k+d)$ is bounded, which depends on both $\mathscr{E}_r(k)$ and $\mathscr{E}_{\hat{o}}(k)$. A solution to this problem will be presented within the section. Before, let us present some important theoretical preliminaries in the sequence. 

\subsection{Theoretical preliminaries}

% Let us recall the following definition (see, for example,  \citet{Khalil_2002}). 
\begin{defn}\citep{Khalil_2002} \label{def:Khalil}
A mapping $H : l_2^{loc} (\integers) \rightarrow l_2^{loc} (\integers)$ is finite-gain $l_2^{loc} (\integers)$-stable if there exist nonnegative constants $\gamma$ and $\beta$ such that \begin{equation}
    \sqrt{{\sum}_{\tau=0}^{k} \|H(u(\tau))\|^{2}} \leq \gamma \sqrt{{\sum}_{\tau=0}^{k} \|u(\tau)\|^2} + \beta
\end{equation}
\noindent for all $u(k) \in l_2^{loc}{\integers}$ and $k \in [0,\infty)$.
\end{defn}
Definition \ref{def:Khalil} is important as it will be helpful to establish the main lemma in this section. To this end, 
from \eqref{observer} and \eqref{eq:Eta}, and by defining the variable $e_\eta(k)=\eta(k)-\hat{\eta}(k)$, let us consider the following system
\begin{equation}
\begin{cases}\label{map:deltawtoOe}
e_\eta(k+1) = (\xbar{A}-L \xbar{C})e_\eta(k) + \xbar{B}_w \Delta^{r+1}w(k), \\
\mathscr{E}_{\hat{o}}(k) = \Gamma(d) e_\eta(k),
\end{cases}
\end{equation}
\noindent which is a mapping from $\Delta^{r+1}w(k)$ to the prediction error $\mathscr{E}_{\hat{o}}(k)$. Next, we present design conditions for the nominal case (i.e., $\lambda=0$).

\subsection{Nominal observer-predictor design}\label{sec:designnominal}
%Let us consider the following lemma for the design of the observer gain $L$.

\begin{lem}\label{lem:design}
Let there exist matrices $P$ in $\mathbb{S}_n^{+}$, $W$ in $ \reals^{n \times m_y}$, and a scalar $\xbar{\gamma} > 0$ such that
\begingroup
\setlength\arraycolsep{5pt}
\begin{equation}\label{LMI:teo}
    \begin{bmatrix}
    P - \Gamma^\T(d)\Gamma(d) & 0 & \xbar{A}^\T P - \xbar{C}^{\T}W^{\T} \\
    \star & \xbar{\gamma} I & \xbar{B}^{\T}_w P \\
    \star & \star & P
    \end{bmatrix} \succ 0.
\end{equation}
\endgroup
\noindent Then, for any initial condition $e_\eta(0)$ and for the observer gain given by $L=P^{-1}W$, the following statements hold:

\begin{enumerate}
    \item The mapping $\Delta^{r+1}w(k) \mapsto \mathscr{E}_{\hat{o}}(k)$ from \eqref{map:deltawtoOe} is $l_2^{loc} (\integers)$-stable with $l_2$ gain less than or equal to $\gamma=\sqrt{\xbar{\gamma}}$.
    \item The error $\mathscr{E}_{\hat{o}}(k)$ is $l_2$-bounded for all $k \in [0,\infty)$ by
    \begin{equation}
    % \sqrt{{\sum}_{\tau=0}^{k} \|O_e(\tau)\|^2}
    \|\mathscr{E}_{\hat{o}}\|_{l_2^{loc}} \leq  \gamma \delta\sqrt{k+1} +  \sqrt{e_\eta^\T(0) P e_\eta(0)}. 
    \end{equation}
    \item When $\Delta^{r+1}w(k)=0$, the error $\mathscr{E}_{\hat{o}}(k)$ converges asymptotically to zero.
    \item There exists a positive constant $\alpha$ and a finite sample time $k_1 \geq 0$ such that the state $e_{\eta}(k)$ is ultimately bounded by $\|e_{\eta}(k)\| \leq \alpha \|\xbar{B}_w \| \delta = \alpha \delta$, $\forall k\geq k_1$.
\end{enumerate}
\end{lem}

\begin{proof}
Consider a quadratic Lyapunov function with $V(e_\eta(k))=e^{\T}_\eta (k) P e_\eta(k)$, $P$ in $\mathbb{S}_n^{+}$. Then, if $V(e_\eta(k))<0$ is ensured along the trajectories of \eqref{map:deltawtoOe}, its asymptotic stability is guaranteed. Now consider LMI \eqref{LMI:teo}. Replace $W$ by $P L$, $\xbar{\gamma}$ by $\gamma^2$, and apply left and right multiplication by \textit{diag}($I,I,P^{-1}$), followed by a Schur complement. Next, apply left and right multiplication by the vector $\begin{bmatrix} e^{\T}_{\eta}(k) & \Delta^{r+1}w^{\T}(k) \end{bmatrix}$ and its transpose, respectively, to find the inequality $\mathcal{K}(k)<0$ where
\begin{equation}\label{eq:l2equation}
    \mathcal{K}(k) = \Delta V(e_\eta(k)) + \|\mathscr{E}_{\hat{o}}(k)\|^{2}-\gamma^2 \|\Delta^{r+1}w(k)\|^{2}.
\end{equation} Then, by computing $\sum_{\tau=0}^{k} \mathcal{K}(\tau) < 0$, taking into account that ${V (e_\eta(k))}>0$, applying square roots to the obtained inequality, and then using the fact that $\sqrt{a^2+b^2} \leq a+b$ for $a,b \in \reals^{+}$, one obtains 
\begin{equation}\label{eq:finalproofteo}
   \|\mathscr{E}_{\hat{o}}\|_{l_2^{loc}}<\gamma
      \|\Delta^{r+1}w\|_{l_2^{loc}} +\sqrt{V (e_\eta(0))}.
\end{equation}
\noindent Thus, by Definition \ref{def:Khalil}, the mapping $\Delta^{r+1}w(k) \mapsto \mathscr{E}_{\hat{o}}(k)$ from \eqref{map:deltawtoOe} is $l_2^{loc} (\integers)$-stable with $l_2$ gain less than or equal to $\gamma$ and bias term $\beta=\sqrt{V (e_\eta(0))}$, which proofs item 1 in Lemma \ref{lem:design}. Item 2 comes directly from Assumption \ref{assump:disturbance} and \eqref{eq:finalproofteo}. For item 3, note that when $\Delta^{r+1}w(k)=0$, relation $\mathcal{K}(k)<0$ implies $\Delta V(e_\eta(k))<-\|\mathscr{E}_{\hat{o}}(k)\|^{2}<0$, therefore the trajectories of $e_\eta(k)$, and consequently of $\mathscr{E}_{\hat{o}}(k)=\Gamma(d) e_\eta(k)$, asymptotically converge to zero. Finally, item 4 is a mere consequence of the asymptotic stability of the linear error system \eqref{map:deltawtoOe}. Thus, all items in Lemma \ref{lem:design} are proven and the proof is complete. \end{proof}

One way to minimize the $l_2$ gain of the mapping $\Delta^{r+1}w(k) \mapsto \mathscr{E}_{\hat{o}}(k)$ and therefore minimize the energy of the prediction error is to minimize $\xbar{\gamma}$ while solving LMI \eqref{LMI:teo}. The following optimization problem can then be formulated in order to achieve better results of the predictor in case of Lemma \ref{lem:design}
\begin{align}\label{eq:opt1}
    \min_{\{P,~W,~\xbar{\gamma}\}} \quad  \xbar{\gamma} \quad \textrm{subject to} \quad \text{\eqref{LMI:teo}.}
\end{align}
\subsection{Discussion}\label{subsec:disc}
Two important results related to the new predictive scheme have been achieved so far in this paper. In another words, for a given plant \eqref{eq:plant} with $\lambda=0$ and taking into account Assumption \ref{assump:disturbance}, we have shown that the proposed predictive scheme leads to a prediction with the following characteristics
\begin{enumerate}[(i)]
\item For all $k \in [0,\infty)$, the $l_2$ norm of the prediction error $x(k+d)-\hat{x}_4(k+d)$ is bounded such as
    \begin{equation}\label{eq:boundpredictionerrror}
    \begin{split}
     &\sqrt{\displaystyle \sum_{\tau=0}^{k} \|x(\tau+d)-\hat{x}_4(\tau+d)\|^2} \\
     & \leq (\gamma+d\mu) \delta\sqrt{k+1} +  \sqrt{e_\eta^\T(0) P e_\eta(0)}.
     \end{split}
    \end{equation}
    \item Such an error can be minimized by running optimization problem \eqref{eq:opt1}.
\end{enumerate}
Moreover, the following proposition holds. 
\begin{prop}\label{prop:nullerror}
The scheme proposed in this paper guarantees null steady-state prediction error in the case of $\Delta^{r+1}w(k)=0$. Moreover, the prediction error is ultimately bounded by
\begin{equation}\label{eq:ultimateboundnessprediction}
    \|\mathscr{E}_{\hat{o}}(k)+\mathscr{E}_r(k)\| \leq \left( \alpha \|\Gamma(d)\| +  d \mu \right) \delta, \forall k \geq k_1. 
\end{equation}
\end{prop}
\vspace{-0.7cm}
\begin{proof}
In case $\Delta^{r+1}w(k)=0$, error $\mathscr{E}_r(k)$ in \eqref{eq:Or} is null for all $k \in [0,\infty)$. Moreover, according to item 3 of Lemma \ref{lem:design}, the error due to observation $\mathscr{E}_{\hat{o}}(k)$ converges asymptotically to zero, thus implying that the prediction error \eqref{eq:predictionerror} also converges asymptotically to zero. Equation \eqref{eq:ultimateboundnessprediction} is a direct consequence of item 4 of Lemma \ref{lem:design} and equation \eqref{eq:boundinOr}. 
\end{proof}

Proposition \ref{prop:nullerror} implies that constant disturbances are perfectly compensated by the prediction scheme since $\delta=0$ for all $r \geq 0$ in this case. In the case of ramp-like disturbances, null prediction error is achieved for $r \geq 1$. More generally, let $n_r$ be the degree of an unknown time-varying polynomial disturbance $w(k) = w_0 + w_1 k + \dots + w_{n_r} k^{n_r}$, then null steady-state prediction error is achieved whenever $r \geq n_r$ since $\Delta^{r+1}w(k)=0$ in this case. Although one can enlarge the family of time-varying disturbances for which null prediction error is achieved at steady-state by increasing the parameter $r$, it should be noted that complexity also increases due to the order of the observer state being augmented, possibly leading to less efficiency in solving optimization problem \eqref{eq:opt1}. 

\section{Disturbance attenuation and robust design}

In this section, we analyse in detail the disturbance attenuation properties of the proposed predictive-based control strategy. Moreover, an iterative algorithm for the robust stabilization of the closed loop is proposed. Thus, by the end of this section, all the goals of the paper will be achieved.

\subsection{Active disturbance rejection analysis}
Since in this work we observe the finite differences of the disturbance, we can compute the estimated value of the future disturbance $w(k+d)$. From \eqref{eq:sum_w}, the following holds
\begin{equation}
    w(k+d)= \hspace{-0.1cm}\sum_{m=0}^{r} \binom{d}{m} \Delta^m w(k)+\hspace{-0.3cm}\sum_{m=r+1}^{\infty} \hspace{-0.15cm}\binom{d}{m} \Delta^m w(k),
\end{equation}
Then, we define the prediction of $w(k+d)$ below 
\vspace{-0.2cm}
\begin{subequations}
\begin{gather}
    \label{eq:predictw}
    \hat{w}(k+d)= \sum_{m=0}^{r} \binom{d}{m} \Delta^m \hat{w}(k)=H(d) \hat{\eta}(k), \\
    \hspace{-0.15cm}H(d) \hspace{-0.12cm}=\hspace{-0.12cm} \begin{bmatrix}
    0_{q \times n_p} & H_w(d)
    \end{bmatrix}\hspace{-0.05cm}, H_w(d) \hspace{-0.12cm}=\hspace{-0.12cm} \begin{bmatrix}
    I_q \binom{d}{0} & \dots & I_q \binom{d}{r}
    \end{bmatrix}\hspace{-0.05cm}. 
\end{gather}
\end{subequations}
In this work, we apply the following modified control law
\begin{equation}\label{eq:activecontrollaw}
    u(k) = K \hat{x}_4(k+d) -\hat{w}(k+d) 
\end{equation}
\noindent with $\hat{x}_4(k+d)$ given by \eqref{eq:prediction} and $\hat{w}(k+d)$ by \eqref{eq:predictw}. 
% Let us define the following disturbance prediction error below and the following two propositions.
% \begin{equation}
% \begin{split}
%     \mathscr{E}_w(k) & = w(k+d)-\hat{w}(k+d)\\
%     & = H(d) e_{\eta}(k)+\sum_{m=r+1}^{\infty} \binom{d}{m} \Delta^m w(k).
% \end{split}
% \end{equation}
% \begin{prop}
% \vspace{-0.6cm}
% The error $\mathscr{E}_w(k)$ is ultimately bounded by $ \|\mathscr{E}_w(k)\| \leq \alpha \|H(d)\| \|\xbar{B}_w \| \delta + \delta \phi_0, \forall k \geq k_1$, where $\phi_j$ has been defined in Proposition \ref{propOr}. Furthermore, it converges asymptotically to zero if $\Delta^{r+1}w(k)=0$.
% \end{prop}
% \begin{proof}
% The proof is similar to the one for the state prediction error $\mathscr{E}_{\hat{o}}(k)+\mathscr{E}_r(k)$ and thus is omitted. \end{proof}
To further analyse the characteristics of the predictive control method proposed in this paper, we consider the reduction method from \cite{Artstein_1982}. For the classical prediction variable $z_1(k) \triangleq \hat{x}_1(k+d)$, the reduction is derived as 
\begin{subequations}
\begin{gather}
\label{eq:reduc1}
    z_1(k+1) = A z_1 (k) + B u(k) + A^d B w(k). \\
    \intertext{Similarly, the reduction for the prediction variables $z_2(k) \triangleq \hat{x}_2(k+d)$ and $z_3(k) \triangleq \hat{x}_3(k+d)$ from \citet{Wu_2021} yield, respectively}
    \label{eq:reduc2}
    \hspace{-0.2cm} z_2(k+1) \hspace{-0.1cm}= \hspace{-0.1cm}A z_2 (k) \hspace{-0.1cm}+\hspace{-0.1cm} B (u(k) \hspace{-0.1cm}+\hspace{-0.1cm} w(k))\hspace{-0.1cm} + \hspace{-0.1cm}A^d B \nabla_d w(k), \\
    \label{eq:reduc3}
    \begin{split}
        \hspace{-0.25cm}z_3(k&+1) \hspace{-0.05cm} = \hspace{-0.05cm} A z_3 (k) \hspace{-0.05cm}+\hspace{-0.05cm} B u(k) + B w(k) + B \nabla_d w(k) \\&+A^d B \nabla_d^2 w(k).
    \end{split}\\ 
    \intertext{Finally, with the new prediction variable $z_4(k) = \hat{x}_4(k+d)$, we obtain the new reduction of \eqref{eq:plant} given below}
    \label{eq:z4}
\begin{split}
    z_4(k+1) &= A z_4 (k) + B u(k) + B \hat{w}(k+d)\\&+A^d \begin{bmatrix}
I & 0
\end{bmatrix} L \xbar{C} e_{\eta}(k).
\end{split} \end{gather}\end{subequations}
\vspace{-0.6cm}
\begin{prop}\label{prop:disturbancerejection}
If $K$ is such that $A+BK$ is a Schur matrix, the control law given by \eqref{eq:activecontrollaw} perfectly cancels polynomial disturbances of order $n_r \leq r$.
\end{prop}
\begin{proof}
Consider \eqref{eq:activecontrollaw} to rewritte \eqref{eq:z4} as $z_4(k+1)=(A+BK)z_4(k)+A^d \begin{bmatrix}
I & 0
\end{bmatrix} L \xbar{C} e_{\eta}(k)$.
From Proposition \ref{prop:nullerror}, the prediction error converges asymptotically to zero when $r \geq n_r$, where $n_r$ is the disturbance degree. Now to demonstrate that $x \rightarrow 0$ it suffices to show that $z_4(k) \rightarrow 0$, which is true since $e_{\eta}(k) \rightarrow 0$ when $n_r \leq r$. \end{proof}

% Proposition \ref{prop:disturbancerejection} already highlights a great improvement concerning the predictive strategy from \citet{Wu_2021}, which is only able to reject step-like disturbances. Next, we will consider analysis for general time-varying disturbances.

Considering the standard control law $u(k)=K x_i (k+d)$, $i=1,2,3$ for the reduced systems \eqref{eq:reduc1}-\eqref{eq:reduc3} and control law \eqref{eq:activecontrollaw} for \eqref{eq:z4}, all four reduced systems can be written in the form $\mathcal{Z}(k+1) = (A+BK)\mathcal{Z}(k)+g(k)$, where $(A+BK)$ is a Schur matrix, which implies global geometric stability of the origin of the nominal system ($g(k) \equiv 0$), meaning that if $\|g(k)\| \leq \xbar{g}$, $\forall k \leq 0$, then there exists a constant $\rho > 0$ and sample time $k_1>0$ such that the ultimate bound given by $\|\mathcal{Z}(k)\| \leq \rho \xbar{g}$, $\forall k \geq k_1$ holds. Recalling the assumption from \citet{Wu_2021} on the disturbance being bounded such as $\| \nabla_d^i w(k) \|\leq d^i D_i < \infty, \forall k \in [id,+\infty)$ for $i \in \{1,2,3\}$, the following ultimate bounds are obtained 
\begin{subequations}
\begin{gather}
    \|z_1(k)\| \leq {\rho} \|B\| \|A^d\| D_0 \\
    \|z_2(k)\| \leq {\rho} \|B\| \left( \| A^d\| d D_1 + D_0\right) \\
    \|z_3(k)\| \leq {\rho} \|B\| \left( \|A^d\| d^2 D_2 + d D_1 + D_0 \right) \\
    \|z_4(k)\| \leq {\rho} \alpha \|A^d\| \|L \xbar{C}\| \delta
\end{gather}
\end{subequations}
\noindent for $k \geq k_1$. Considering $\varphi = \sum_{j=0}^{d-1} \|A^j\|$, \citet{Wu_2021} demonstrated that the expressions $\|x(k)\| \leq \| \hat{x}_1(k)\|+\varphi \|B\| D_0$, $\|x(k)\| \leq \| \hat{x}_2(k)\|+\varphi \|B\| d D_1$, and $\|x(k)\| \leq \| \hat{x}_3(k)\|+\varphi \|B\| d^2 D_2$ hold. Furthermore, with the proposed prediction scheme, one obtains $\|x(k)\| \leq \| \hat{x}_4(k)\|+\left( \alpha \|\Gamma(d)\| +  d \mu \right) \delta$, leading to the following ultimate bounds on the state (where $x_\infty \triangleq \|x(k)\|, k\rightarrow \infty$)
\begin{subequations}
\begin{gather}
\label{eq:ultimatex1}
    x_\infty \leq \left(\varphi + {\rho} \|A^d\| \right) \|B\| D_0 \triangleq c_1\\
    \label{eq:ultimatex2}
   x_\infty \leq \left(\varphi + {\rho} \|A^d\| \right) \|B\| d D_1 + {\rho} \|B\| D_0 \triangleq c_2 \\
   \label{eq:ultimatex3}
    \hspace{-0.25cm}x_\infty\hspace{-0.1cm} \leq \hspace{-0.1cm}\left(\varphi \hspace{-0.08cm}+\hspace{-0.08cm} {\rho} \|A^d\| \right)\hspace{-0.1cm} \|B\| d^2 D_2 \hspace{-0.05cm}  + \hspace{-0.05cm} {\rho} \|B\|\hspace{-0.1cm} \left( d D_1 \hspace{-0.05cm}+\hspace{-0.05cm}D_0 \right) \hspace{-0.08cm} \triangleq \hspace{-0.08cm} c_3\\
    \label{eq:ultimatex4}
   x_\infty \leq \left({\rho} \alpha \|A^d\| \|L \xbar{C}\|+ \alpha \|\Gamma(d)\| +  d \mu \right) \delta \triangleq c_4
\end{gather}
\end{subequations}
\begin{teo}
\vspace{-0.8cm}
The proposed prediction scheme with $\hat{x}_4$ can always deliver enhanced attenuation of time-varying polynomial disturbances compared to the classical prediction with $\hat{x}_1$ and the two predictions $\hat{x}_2$ and $\hat{x}_3$.   
\end{teo}
\begin{proof}
The proof is straightforward by noting that perfect disturbance attenuation ($x_\infty=0$) with the proposed predictor variable $\hat{x}_4$ can be achieved by choosing $r \geq n_r$ (since $\delta=0$ in this case), which is not achievable with the prediction variables $\hat{x}_1$, $\hat{x}_2$, and $\hat{x}_3$.  
\end{proof}
\begin{rem}
If $\delta \leq d^i D_i, i \in \{1,2,3\}$, a sufficient condition to obtain $c_4 \leq c_i, i \in \{1,2,3\}$, is given by $\left({\rho} \alpha \|A^d\| \|L \xbar{C}\|+ \alpha \|\Gamma(d)\| \right) \leq \left(\varphi + {\rho} \|A^d\| \right)$, whose attainability depends on $\alpha$ and therefore on the quality of design of the observer parameter. This condition is obtained by noting (from Proposition \ref{propOr}) that the influence of the term $d \mu$ in \eqref{eq:ultimatex4} can always be eliminated by choosing $r>d-2$. For the special case of sinusoidal disturbances $w(k)=D_0 sin (f_0 k)$, note that increasing the parameter $r$ can lead to arbitrarily small $\delta$, which might lead to tighter bounds in \eqref{eq:ultimatex4}.    
\end{rem}

\subsection{Robust Lyapunov design}\label{sec:dist}

It is easy to check that in the nominal case (system without uncertainties), any gains $K$ and $L$ such that $A+BK$ and $\xbar{A}-L \xbar{C}$ are Schur matrices yield a stable closed-loop system. This is not true in the uncertain case. In this section, we will conclude the goals of this paper by providing a solution to the robust stabilization of the predictor-based closed-loop system. First, a backstepping transformation on the control law is applied, and then a representation of the closed-loop allowing the gathering of design conditions based on an iterative algorithm is presented. Consider the backstepping transformation below, which is a discrete-time version of the one in \cite[p. 37]{krstic2009}
\begin{gather}\label{eq:transf.}
    \hspace{-0.3cm}v(\theta)=u(\theta)-K[A^{\theta+k+d}\hat{x}(k) + \hspace{-0.35cm}\sum_{j=1}^{\theta-k+d}\hspace{-0.3cm} A^{j-1}Bu(k-j)],
\end{gather}
where $\theta \in [k-d,k],k \geq 0$. Consider the extended variable $\zeta(k)=[\hat{x}^\T(k)~~e_\eta^\T(k)]^\T \in \reals^{n_\zeta}$, $n_\zeta=n+n_p$, the control $u(k)$ given in \eqref{eq:activecontrollaw}, and the facts that $u(k-d)=v(k-d)+K\hat{x}(k)$ and $\bm{\hat{w}}(k)=\bm{w}(k)-\begin{bmatrix}
0 & I
\end{bmatrix} e_{\eta}(k)$. Then, the following closed-loop representation is obtained for the uncertain case, i.e, \eqref{eq:plant} with $\lambda \neq 0$
\begin{align} \label{eq:closedloop}
\begin{cases}
    \zeta(k+1) = \xbar{A}_\zeta \zeta(k)
    + \xbar{B}_{\zeta}v(k-d)+\xbar{B}_{w\zeta} \bm{w}(k),\\
    y(k)=C_\zeta \zeta(k) + D_\zeta \bm{w}(k)
\end{cases} 
\end{align}
\vspace{-0.6cm}
\begin{align*}
     &\xbar{A}_\zeta(k) = A_{\zeta}+E_{\zeta} \Omega(k) F_{A \zeta}, ~\xbar{B}_{\zeta}(k) = B_{\zeta}+E_{\zeta} \Omega(k) F_{B \zeta},\\
     &\xbar{B}_{w\zeta}(k)=B_{w\zeta}+E_{\zeta} \Omega(k) F_{w \zeta}, ~C_\zeta=C\begin{bmatrix}
I_{n_p} & I_{n_p} & 0
\end{bmatrix},\\
       &\left[\begin{array}{c|c}
        \hspace{-0.1cm} A_{\zeta} & B_{\zeta}
     \end{array} \hspace{-0.1cm} \right]=
     \left[\begin{array}{cc|cc}
    A+BK & [I~~0]L\xbar{C} & B  \\ 0 & \xbar{A}-L \xbar{C} & 0 
    \end{array}\right], E_{\zeta} = \begin{bmatrix}
    0 \\ \xbar{E}  
    \end{bmatrix}, \\
 &F_{A\zeta} = \begin{bmatrix} 
    F_A+F_B K & F_A[I~~0] 
    \end{bmatrix}, F_{B\zeta}
    =F_B,\\
    &F_{w \zeta} = F_B B_{w1},\xbar{E} = \begin{bmatrix}
  \lambda E \\ 0
\end{bmatrix},B_{w \zeta}=\begin{bmatrix}
    B B_{w1}\\
    \xbar{B}_w B_{wn1}
    \end{bmatrix}, \\
     &B_{w1}=\begin{bmatrix}
I_q & 0_{q \times qr}
\end{bmatrix},B_{wn1}=\begin{bmatrix}
 0_{q \times qr} & I_q
\end{bmatrix}, D_\zeta=D_w B_{w1}.
    \end{align*}

\begin{teo} \label{designtheorem}
Given a delay bound $d$, assume that there exist matrices ${P} $ in $\mathbb{S}_{n_\zeta}^{+}$, ${Z}$ in $\mathbb{S}_{m_u}^{+}$, ${\xbar{P}} $ in $\mathbb{S}_{n_\zeta}^{+}$, ${\xbar{Z}}$ in $\mathbb{S}_{m_u}^{+}$, ${K}$ in $\mathbb{R}_{m_u \times n_p}$, ${L}$ in $\mathbb{R}_{n \times m_y}$, and positive scalars ${\xbar{\gamma}}$ and $\epsilon$, such that the inequality
\begin{equation} \label{eq:LMIrobust}
    \begin{bmatrix}
    I_o & Y^\T & X \\
    \star & \epsilon I &  0 \\
    \star &  \star &   \epsilon I
    \end{bmatrix}  \succ 0,
\end{equation}
where $X=\begin{bmatrix}
   0 & 0 & -E_{\zeta}^\T & 0
   \end{bmatrix}^\T$, $Y=\begin{bmatrix}
   F_{A_{\zeta}} & F_{B_{\zeta}} & F_{w_{\zeta}} & 0
   \end{bmatrix}$,
\begin{align*}\hspace{-0.2cm}I_o =    \begin{bmatrix}
      P-C_{\zeta}^\T C_{\zeta}  &  0  &  -C_{\zeta}^\T D_{\zeta}  &  A_{\zeta}^\T  & \sqrt{d} L_e^\T \xbar{K}^\T\hspace{-0.1cm}(d)\\
      \star &  Z &  0  & B_{\zeta}^\T  &  0\\
      \star &  \star  &  \xbar{\gamma} I  - D_{\zeta}^\T D_{\zeta}   &  B_{w\zeta}^\T &  -\sqrt{d} ~\xbar{K}^\T \hspace{-0.1cm}(d)\\
      \star  &  \star  & \star  & \xbar{P}  &   0\\
      \star &  \star  & \star  & \star  & \xbar{Z}
        \end{bmatrix}\hspace{-0.1cm},
\end{align*}
and $~\xbar{K}(d)=H_w(d)-K \Gamma(d)$ holds subject to the equality constraints
\begin{equation} \label{eq:constraints} P \xbar{P}=I \text{ and } Z \xbar{Z}=I.
\end{equation}
Then, the closed-loop system \eqref{eq:closedloop} is robustly stable with \(l_2\) gain less than or equal to \(\gamma=\sqrt{\xbar{\gamma}}\) and a guaranteed level of robustness~given~by~$\lambda$.
\end{teo}
\begin{proof}
Consider the LKF given below
\begin{equation}
\hspace{-0.1cm} V(k)=\zeta^\T \hspace{-0.1cm}(k) P \zeta(k)+ \hspace{-0.2cm} \sum_{\theta=k-d}^{k-1}\hspace{-0.1cm} (1+\theta+d-k)v^\T \hspace{-0.1cm}(\theta) Z v(\theta)   \hspace{-0.05cm}
\end{equation}
where $P$ in $\mathbb{S}_{n_\zeta}^{+}$ and $Z$ in $\mathbb{S}_{m_u}^{+}$. Then, if \(\Theta(k) = \Delta V(k) + y^\T(k) y(k) - \gamma^2 \bm{w}^\T(k) \bm{w}(k) \leq 0\) is ensured along the trajectories of \eqref{eq:closedloop}, its asymptotic stability is guaranteed with an $l_2$-gain performance $\gamma$. From the transformation \eqref{eq:transf.} and \eqref{eq:activecontrollaw}, it holds that $v(k)=K T(d) \bm{\hat{w}}(k)-H(d)\hat{\eta}(k)$, which can be rewritten as $v(k) = \left( H_w(d)-K T(d) \right) \left( L_{1e} e_{\eta}(k) - \bm{w}(k) \right)$,
where $L_{1e}=[0_{q(r+1)\times n_p}~~I_{q(r+1)}]$. Considering the expression for $v(k)$, $e_\eta (k)=L_{2e} \zeta(k),~L_{2e}=[0_{n \times n_p}~~I_n]$, and the extended vector $\mu(k)=[\zeta^\T (k)~~v^\T (k-d)~~\bm{w}^\T(k)]^\T$, we obtain the bound $\Theta(k) \leq \mu^\T(k) \Phi (d) \mu (k)$, where
\begin{equation}\label{eq:lmianalise}
  \Phi(d)=\begin{bmatrix}
  \phi_{11}  &  \xbar{A}_\zeta^\T  \hspace{-0.1cm}  P \xbar{B}_\zeta &  \phi_{13}\\
  \star   &  \xbar{B}_\zeta^\T \hspace{-0.1cm}  P \xbar{B}_\zeta \hspace{-0.1cm} -\hspace{-0.1cm} Z  &  \xbar{B}_\zeta^\T \hspace{-0.1cm}  P \xbar{B}_{w\zeta}\\
  \star  &  \star  &   \phi_{33}
  \end{bmatrix} \preceq 0,
\end{equation}
\vspace{-0.5cm}
\begin{align*}
&\phi_{11}=\xbar{A}_\zeta^\T P \xbar{A}_\zeta-P+ L_e^\T \xbar{K}^\T (d) Z \xbar{K}(d) L_e d + C_\zeta^\T C_\zeta,\\
&\phi_{13}=C_\zeta^\T D_\zeta + \xbar{A}_\zeta^\T P \xbar{B}_{w\zeta}-L_e^\T \xbar{K}^\T (d) Z \xbar{K}(d)d,\\
&\phi_{33}=\xbar{B}_{w\zeta}^\T P \xbar{B}_{w\zeta}+ \xbar{K}^\T(d)Z\xbar{K}(d)d+D_{\zeta}^\T D_\zeta-\gamma^2I,\\
&\xbar{A}_\zeta^\T=A_\zeta+E_\zeta \Omega(k)F_{A\zeta},~\xbar{B}_\zeta^\T=B_\zeta+E_\zeta \Omega(k)F_{B\zeta},\\
&\xbar{B}_{w\zeta}^\T=B_{w\zeta}+E_\zeta \Omega(k)F_{w\zeta},~L_e=L_{1e}L_{2e}.        
\end{align*}
Applying Schur complement in \eqref{eq:lmianalise}, followed by changes of variable $\xbar{\gamma}=\gamma^2,~\xbar{P}=P^{-1},~\xbar{Z}=Z^{-1}$, we obtain
\begin{equation}\label{eq:teo_separated}
-I_o+X\Omega(k)Y+Y^\T \Omega^\T(k) X^\T \prec 0
\end{equation}
with 
$I_o,~X,$ and $Y$ given in Theorem \ref{designtheorem}. Since $\Omega^\T(k)\Omega(k) \leq I$, the following inequality holds for any scalar $\epsilon>0$ \citep{Gu2003}
\begin{equation}
    X\Omega(k)Y+Y^\T \Omega^\T(k) X^\T \leq \epsilon XX^\T + \epsilon^{-1}Y^\T Y.
\end{equation}
Then $-I_o + \epsilon XX^\T + \epsilon^{-1}Y^\T Y \prec 0$
% \begin{equation}\label{eq:teoepsilon}
%     -I_o + \epsilon XX^\T + \epsilon^{-1}Y^\T Y \prec 0
% \end{equation}
is a sufficient condition to fulfill \eqref{eq:teo_separated}. Applying Schur complement in this last inequality, we arrive in condition \eqref{eq:LMIrobust} of Theorem \ref{designtheorem}.
\end{proof}

\begin{rem}
Stability of the system with the original variables $u(k),x(k)$ is guaranteed. The LKF for the original variables, which is far from simple, can be explicitly written by taking the inverse of the transformation of \eqref{eq:transf.} \cite[p. 38]{krstic2009} and using the fact that $x(k)=\hat{x}(k)+\begin{bmatrix}
I & 0
\end{bmatrix}e_{\eta}(k)$. The backstepping approach, thus, provides a more elegant and simpler way inspired by partial differential equations (PDEs) to provide stability, as highlighted in \cite{krstic2009}. 
\end{rem}

\subsection{CLL Algorithm} \label{iterativealgorithm}

The cone complementary linearization (CCL) algorithm \citep{Ghaoui_1997} is applied for solving the conditions in Theorem \ref{designtheorem}. First, we relax the equality constraints from \eqref{eq:constraints} with the following LMI conditions
\begin{equation}\label{constraints}
    \begin{bmatrix}
    P &  I_{n_\zeta}\\
    I_{n_\zeta}  & \xbar{P}
    \end{bmatrix} \succeq 0,~ \begin{bmatrix}
    Z &  I_{m_u}\\
    I_{m_u}  & \xbar{Z}
    \end{bmatrix} \succeq 0
\end{equation}
subject to the minimization of the objective function
\begin{equation*}
    trace(P\xbar{P}+\xbar{P}P)+trace(Z\xbar{Z}+\xbar{Z}Z).
\end{equation*}
A general form for the iterative algorithm used in this paper is presented below. 

\textbf{Part I -- Delay maximization}
\begin{itemize}
    \item Step 1: Given $K$ and $L$ such that ($A+BK$) and ($\xbar{A}-L\xbar{C}$) are Schur matrices, choose a small value for $d$, set $\lambda=0$, and find a solution for LMI \eqref{eq:lmianalise}. Then, set $i=1,k=1$ and
    \begin{equation*}
        d_0=d+n_d,P_0=P,\xbar{P}_0=P^{-1},Z_0=Z,\xbar{Z}_0=Z^{-1}.
    \end{equation*}
    being $n_d$ an incremental value for each iteration.
    \item Step 2: Set $d_i=d_{i-1}+n_d,P_k=P_{k-1},\xbar{P}_k=P_{k-1}^{-1},Z_k=Z_{k-1},\xbar{Z}_k=Z_{k-1}^{-1}$ and solve LMIs \eqref{eq:LMIrobust}-\eqref{constraints} subject to:
    \begin{equation*}
       minimize~~ trace(P\xbar{P}_k+\xbar{P}P_k)+trace(Z\xbar{Z}_k+\xbar{Z}Z_k),
    \end{equation*}
    for all $k \leq N^o_{iterations}$, $k \in \mathbb{I}_{+}$, being $N^o_{iterations}$ a chosen number of iterations.
    \item Step 3: If $d_i < d_{max}$, set $i=i+1$ and go back to step 2. If a feasible solution is found for $d_i=d_{max}$, then go to \textbf{Part II}. If not, set a smaller $d_{max}$ and go back to Step 2.
    \end{itemize}
\textbf{Part II -- Robustness maximization}    
    \begin{itemize}
    \item Since we have a $P_k$ and a $Z_k$ from \textbf{Part I}, set $P_0=P_k,\xbar{P}_0=P_k^{-1},Z_0=Z_k,\xbar{Z}_0=Z_k^{-1}$. After that, the steps in this part are similar to Step 2 and Step 3 from \textbf{Part I}. However, instead of doing an increment on the system delay, we have a fixed $d=d_{max}$ and make small increments on the uncertainty parameter $\lambda$, such that $\lambda_i=\lambda_{i-1}+n_\lambda$, being $n_\lambda$ an incremental value for each iteration.
     \end{itemize}   
     
     \textbf{Part III -- $l_2$-gain minimization}    
    \begin{itemize}
    % \item Since we have a $P_k$ and a $Z_k$ from \textbf{Part II}, set $P_0=P_k,\xbar{P}_0=P_k^{-1},Z_0=Z_k,\xbar{Z}_0=Z_k^{-1}$ . After that, the steps in this part are similar to \textbf{Part II}. Here, we have fix $d=d_{max}$ and $\lambda=\lambda_{max}$ and do a decrement on $\xbar{\gamma}$, such that $\xbar{\gamma}_i=\xbar{\gamma}_{i-1}-n_{\xbar{\gamma}}$, being $n_{\xbar{\gamma}}$ a decremental value for each iteration and $\xbar{\gamma}_0=\xbar{\gamma}-n_{\xbar{\gamma}}$.
    \item Similarly to \textbf{Part II}, we repeat the steps from \textbf{Part I}, but this time the delay and the robustness level are both already fixed as $d=d_{max}$ and $\lambda=\lambda_{max}$. Then, the $l_2$-gain performance index is minimized by decreasing its value at each iteration such that $\xbar{\gamma}_i=\xbar{\gamma}_{i-1}-n_{\xbar{\gamma}}$, being $n_{\xbar{\gamma}}$ a small positive scalar.
     \end{itemize}   

\section{Numerical examples}
\label{sec:numericalexample}

% In this section, we present comparisons of the proposed prediction scheme with the two predictors from \citet{Wu_2021} and the predictor-based controllers from \citet{HAO2019} and \citet{GONZALEZ2021}. The next subsections present simulations results for both constant and time-varying disturbances. We consider \(r=0\) for the case of constant disturbance and $r=2$ for the time-varying disturbance.

\subsection{Nominal system and constant disturbances}

In order to evaluate the proposed predictor, let us consider the perturbed input-delayed system recently analysed in \citet{Wu_2021}, with the same delay and initial conditions. We utilise \eqref{eq:prediction} with $r=0$ for the prediction and the control law \eqref{eq:activecontrollaw}, starting with $\hat{\eta}(0)=\left[x(0) ~~ w(0)\right]$. The controller gain $K=\left[-3.14 ~~ 1.5\right]$ is the same used in \citet{Wu_2021}. The constant disturbance $w(k)=1.6$, $k \in [0,\infty)$ is also the same from \citet{Wu_2021}. Applying optimization problem \eqref{eq:opt1} we get the observer gain below
\vspace{-0.2cm}
\begin{equation*}
     L = \begin{bmatrix}
    0.0658  & ~\phantom{-}3.1467 & ~-0.0989 \\
    1.0000 & ~-0.7642  & ~\phantom{-}0.6358
    \end{bmatrix}^\T.
    \vspace{-0.2cm}
\end{equation*}
Figure \ref{ex1:relationepsilon} shows the norm of the plant states for the compared schemes. As it can be seen, the proposed approach yields enhanced disturbance attenuation, being able to completely reject the disturbance after about ten samples. On the other hand, while the predictor-based controllers from \citet{Wu_2021} present a good level of disturbance attenuation, it fails to completely reject the disturbance at steady-state. %A better disturbance attenuation is also achieved during the transient phase with the proposed predictor.
\begin{figure}[ttb!]
\center{\includegraphics[width=\linewidth,angle=0]{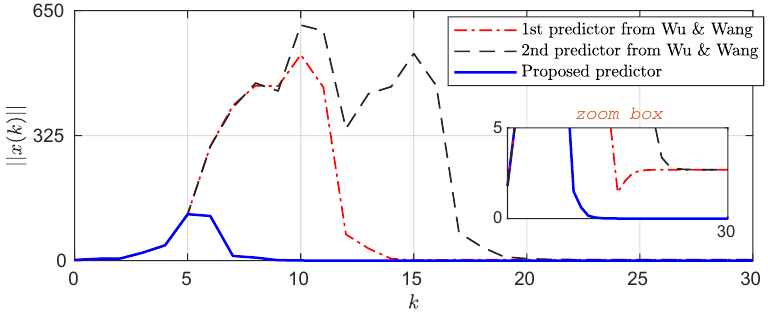}
\caption{Example 1: Norm of the plant states for constant disturbance case.}
\label{ex1:relationepsilon}}
\end{figure}

\subsection{Time-varying uncertainties and disturbances}

Let us consider the uncertain perturbed input-delay system recently analysed in both \cite{GONZALEZ2021} and \cite{HAO2019}. %with matrices
% \begin{equation*}
% \begin{split}
% &A=\begin{bmatrix}
% 0.9505 & 0.1149\\
% -1.0339 & 1.2952
% \end{bmatrix},B=\begin{bmatrix}
% 0.0055\\
% 0.1149
% \end{bmatrix},C=\begin{bmatrix}
% 1 \\0 \end{bmatrix}^\T,\\
% &E=\begin{bmatrix}
% 0 \\ 0.03
% \end{bmatrix}, F^\T_A=\begin{bmatrix}
% 0.1 \\ 0.2
% \end{bmatrix}, F_B=0.05, \Delta(k)=sin(k).
% \end{split}
% \end{equation*}
For a fair comparison, the same parameters from Example 1 of \cite{GONZALEZ2021} are considered. A time-constant delay $d=6$, a sample time $T_s=0.1~s$, an uncertainty index $\lambda=1$, and disturbances generated by the external system presented in \cite{GONZALEZ2021}. We utilize \eqref{eq:prediction} for the prediction with $r=2$ and \eqref{eq:activecontrollaw} for the control law. %starting with $\hat{\eta}^\T(0)=[x(0)~~w(0)]$, $x(0)=[1.5~1]^\T$ and $w(0)=[4.0000~~-0.0036~~0.0346]^\T$. 
Applying the iterative algorithm from section \ref{iterativealgorithm}, we obtain the controller gain $K=[-3.0926~~-5.1147]$, the observer gain $L=[3.5374~~34.8991~~162.1663~~38.0858~~4.0755]^\T$, and a performance index $\gamma=577.33$.
A comparison of the system output norm with others schemes is presented in Figure \ref{ex2}. As it can be seen, the proposed approach shows an improvement on the transient performance with a much faster and better behaved convergence to steady-state, which can be appreciated. Furthermore, considering the same uncertain system with $\lambda=0.1$, we obtain a feasible solution for a bigger delay of $d=27$ with gains $K=[-33.3052  ~~ -6.5767]$ and $L=[1.8527~~7.1026~~1.9133]$, in this case $r=0$. With this same level of robustness, \cite{GONZALEZ2021} achieves a maximum delay of $d=16$. It is fair to say, however, that \cite{GONZALEZ2021} deals with time-varying delays, while we deal with the constant case as in \cite{HAO2019}.  

\begin{figure}[ttb!]
\center{\includegraphics[width=\linewidth,angle=0]{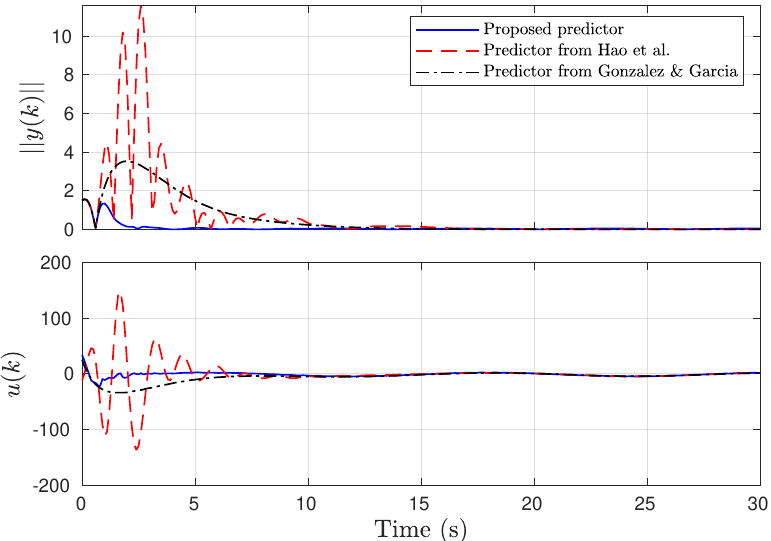}
\caption{Example 2: Norm of the plant output and control signal.}
\label{ex2}}
\end{figure}
% In this case, we consider a sinusoidal disturbance $w(k)=0.6sin\left(\frac{1.35}{2\pi}k\right)$, $k \in [0,\infty)$. Applying optimization problem \eqref{eq:opt1} we get $L = \begin{bmatrix} L_x & L_w\end{bmatrix}^\T$, with $L_x= \begin{bmatrix}
% -0.3901 & 3.1998 \\
% 1.0000 & 1.6621 
% \end{bmatrix}$ and
% \begin{equation*}
%     L_w = \begin{bmatrix}
%      -0.0005 & -0.0006 & -0.0005 & -0.0003 & -0.0001 \\
%      7.8420 & 8.5803 & 5.5770 & 2.0116 & 0.3137
%     \end{bmatrix}^\T
% \end{equation*}
% \noindent for the observer.

% Figure \ref{ex1:sincase} shows the norm of the plant states for the compared predictor schemes. We can verify that the norm reaches high values on the transitory response in both predictors from \citet{Wu_2021}, while on steady state the disturbance is not properly attenuated. On the other hand, the proposed predictor presents improved performance by attenuating the disturbance at steady state and by presenting much lower values for the norm during the transient phase. 

% \begin{figure}[ttb!]
% \center{\includegraphics[width=\linewidth,angle=0]{Figures/normstates_sin.pdf}
% \caption{Example 1: Norm of the plant states for time-varying disturbance case.}
% \label{ex1:sincase}}
% \end{figure}

\section{Conclusions}
\label{sec:conclusions}

This work proposed a newton-series-based method for computing the prediction of discrete-time input-delayed systems in the presence of unknown disturbances and time-varying modelling uncertainties. Effectiveness was demonstrated by means of numerical examples that compared the proposed method against different strategies, showing improvements in both disturbance attenuation and robustness characteristics by achieving stabilization for bigger delay bounds in an example from the recent literature. Future and ongoing research include the development of event-driven observer-predictor-based control methods for infinite-dimensional systems modelled by PDEs and subject to control constraints.

% The method employs extended observers for the disturbances and their finite differences up to order $r$, which are then employed to estimate future values of the disturbance signal, leading to prediction \eqref{eq:prediction}. 

% Effectiveness was demonstrated by means of numerical examples that compared the proposed method against two strategies from the recent literature. From the comparisons, it was clear that the attenuation of the unknown disturbance was enhanced. Analytical expressions were also derived to prove that, under Assumption \ref{assump:disturbance}, the $l_2$ norm of the prediction error is bounded, as given by \eqref{eq:boundpredictionerrror}, for all $k \in [0,\infty)$, while LMI-based design of the observer was developed to successfully minimize such an error by means of optimization procedure \eqref{eq:opt1}. Moreover, we showed that null steady-state prediction error is achieved for a class of unknown time-varying polynomial disturbances whenever $r \geq n_r$, being $n_r$ the order of the disturbance. 

\begingroup
\color{black} \begin{ack}                               % Place acknowledgements
We are sincerely thankful to the anonymous reviewers and the associate editor for their constructive comments.  % here.
\end{ack}
\endgroup

% \newblock{\color{black} \begin{ack}                               % Place acknowledgements
% We are sincerely thankful to the anonymous reviewers and the associate editor for their constructive comments.  % here.
% \end{ack}
% }

\bibliographystyle{apacite}        % Include this if you use bibtex 
\bibliography{autosam}          

%% appendix
\appendix

\section{Proof of Proposition \ref{propOr}} 
\label{appendix:Or}
%Each appendix must have a short title. 
Consider equation \eqref{eq:Or}. By applying the change of variables \(l=m-r-1\), we obtain
\begin{equation}\label{eq:Or_apen}
\hspace{-0.2cm}\mathscr{E}_r(k)\hspace{-0.1cm}=\hspace{-0.2cm}\sum_{j=1}^{d}\hspace{-0.1cm} A^{j-1} B \hspace{-0.1cm} \left(\sum_{l=0}^{\infty} \hspace{-0.1cm} \binom{d-j}{l+r+1} \Delta^{l+r+1}  w(k) \hspace{-0.1cm} \right).
\end{equation}
Now, by recursion, note that:
\begin{align}
       \nonumber \Delta^{1+(r+1)}w(k)&=\Delta (\Delta^{r+1} w(k))\\\nonumber
       &=\Delta^{r+1} w(k+1)-\Delta^{r+1} w(k),\\ \nonumber
       \underset{\vdots}{\Delta^{2+(r+1)}w(k)}&=\underset{\displaystyle \phantom{=}-2\Delta^{r+1} w(k+1)+\Delta^{r+1}w(k)}{\Delta (\Delta (\Delta^{r+1} w(k)))=\Delta^{r+1} w(k+2)}\\ \nonumber
       \Delta^{l+(r+1)}w(k)&=\sum_{i=0}^{l} (-1)^{l-i} \binom{l}{i} \Delta^{r+1} w(k+i).
\end{align}
By substituting the obtained expression for \(\Delta^{l+r+1}w(k)\) into \eqref{eq:Or_apen}, we obtain $\mathscr{E}_r(k)=  \sum_{j=1}^{d} A^{j-1} B W_j$, where 
\begin{equation*}
    W_j = \displaystyle \sum_{l=0}^{\infty} \binom{d-j}{l+r+1} \left(\sum_{i=0}^{l} (-1)^{l-i} \binom{l}{i} \Delta^{r+1} w(k+i) \right),
\end{equation*}
\noindent which leads to
$ \|\mathscr{E}_r(k)\|$ being bounded such as
\begin{equation}
    \|\mathscr{E}_r(k)\|
    \leq \sum_{j=1}^{d} \|Y_j^{\frac{1}{2}} W_j\|,
\end{equation}
where \(Y_j= B^\T {A^{^\T j-1}} A^{j-1}B\). Next, by applying $\sum_{j=1}^{d} \|Y_j^{\frac{1}{2}} W_j\|\leq d \displaystyle \max_{j=\{1 \dots d\}} \|Y_j^{\frac{1}{2}} W_j\|$ and using the fact that $\|Y_j^{\frac{1}{2}} W_j\| \leq \sigma_{max}(Y_j^{\frac{1}{2}})\|W_j\|$, where $\sigma_{max}(Y_j^{\frac{1}{2}})$ denotes the maximum singular value of $Y_j^{\frac{1}{2}}$, we obtain 
\begin{equation}\label{eq:OrWithWj}
    \|\mathscr{E}_r(k)\| \leq d \max_{j=\{1 \dots d\}} \left(\sigma_{max}(Y_j^{\frac{1}{2}})\|W_j\|\right).
\end{equation}
Now, let us find an expression for $\|W_j\|$. From the definition of $W_j$, it follows that
\begin{equation*}
    \|W_j\|\leq \sum_{l=0}^{\infty}  \binom{d-j}{l+r+1} 
    \sum_{i=0}^{l}  \left|(-1)^{l-i} \binom{l}{i}\right| \left\|\Delta^{r+1} w(i)\right\|,
\end{equation*}
\noindent where $|\cdot|$ stands for the absolute value operator and $\Delta^{r+1} w(i)$ is a short notation for $\Delta^{r+1} w(k+i)$. Since \(\|\Delta^{r+1} w(k+i)\| \leq \delta\) for all $i$ (from Assumption \ref{assump:disturbance}) and $\sum_{i=0}^{l} \binom{l}{i}=2^l$, the following expression holds
\begin{equation}
   \|W_j\|\leq \delta \sum_{l=0}^{\infty}  \binom{d-j}{l+r+1} 2^l.
\end{equation}
Furthermore, by taking into account the fact that
\begin{equation*}
    \sum_{l=0}^{\infty}  \binom{d-j}{l+r+1} 2^l=0,~\text{for}~l>d-j-r-1,
\end{equation*}
we find the bound $\|W_j\|\leq \delta \phi_j$, where the term $\phi_j = \displaystyle \sum_{l=0}^{d-j-r-1}  \binom{d-j}{l+r+1} 2^l$ is a finite series. From \eqref{eq:OrWithWj} and $\|W_j\|\leq \delta \phi_j$, we arrive at the expression
\begin{equation}\label{eq:boundinOr}
    \|\mathscr{E}_r(k)\| \leq  \delta d \mu,~~ \forall k \in [0,\infty),
\end{equation}
\noindent where $\mu =\displaystyle \max_{j=\{1 \dots d\}} \left(\sigma_{max}(Y_j^{\frac{1}{2}})\phi_j\right)$. Finally, summing \eqref{eq:boundinOr} from \(0\) to \(k\) and taking square roots yields
\begin{equation}\label{Eq:OrAppendix}
   \|\mathscr{E}_{r}\|_{l_2^{loc}} \leq \sqrt{\sum_{\lambda=0}^{k} \delta^2 d^2 \mu^2}
    = \delta d \mu \sqrt{k+1}, 
\end{equation}
\(\forall k \in [0,\infty)\), thus completing the demonstration of Equation \eqref{eq:propOr}. Moreover, $\phi_j=0$ whenever the upper limit of the sum is negative, i.e. $d-j-r-1<0$. Therefore, if $r>d-2$, which corresponds to $j=1$, it follows that $\phi_j=0$ for all $j=\{1\cdots d\}$. Thus, $\mu=0$ implying $\|\mathscr{E}_{r}\|_{l_2^{loc}}=0$ in \eqref{Eq:OrAppendix}, which completes the proof.  
% \section{Some Latin vocabulary}         % Sections and subsections are supported  
                                        % in the appendices.
\end{document}